\documentclass[twoside]{article}
\usepackage{aistats2012}
\usepackage{amsmath,amsfonts,amssymb}
\usepackage{graphicx}

\newcommand{\Id}{\mathtt{I}}

\newcommand{\spn}[1]{\ensuremath{\mathrm{span}\left(#1\right)}}
\newcommand{\rank}[1]{\ensuremath{\mathrm{rank}\left(#1\right)}}
\newcommand{\norm}[1]{\ensuremath{\left\| #1 \right\|}}

\newtheorem{thm}{Theorem}
\newtheorem{lem}{Lemma}

\newtheorem{defn}{Definition}

\newtheorem{rmk}{Remark}

\begin{document}
\twocolumn[

\aistatstitle{Exact Subspace Segmentation and Outlier Detection by\\ Low-Rank Representation}

\aistatsauthor{ Anonymous Author 1 \And Anonymous Author 2 \And Anonymous Author 3 }

\aistatsaddress{ Unknown Institution 1 \And Unknown Institution 2 \And Unknown Institution 3 } ]

\begin{abstract}
In this work, we address the following matrix recovery problem: suppose we are given a set of data points containing two parts, one part consists of samples drawn from a union of multiple subspaces and the other part consists of outliers. We do not know which data points are outliers, or how many outliers there are. The rank and number of the subspaces are unknown either. Can we detect the outliers and segment the samples into their right subspaces, efficiently and exactly? We utilize a so-called {\em Low-Rank Representation} (LRR) method to solve this problem, and prove that under mild technical conditions, any solution to LRR exactly recovers the row space of the samples and detect the outliers as well. Since the subspace membership is provably determined by the row space, this further implies that LRR can perform exact subspace segmentation and outlier detection, in an efficient way.
\end{abstract}
\section{Introduction}
This paper is about the following problem: suppose we are given a data matrix $X$, each column of which is a data point, and we know it can be decomposed as
\begin{eqnarray}\label{equ.beginproblem}
X = X_0 + C_0,
\end{eqnarray}
where $X_0$ is a low-rank matrix with the column vectors drawn from a union of multiple subspaces, and $C_0$ is a column-sparse matrix that is non-zero in only a fraction of the columns. Except these mild restrictions, both components are arbitrary. In particular we do not know which columns of $C_0$ are non-zero, or how many non-zero columns there are. The rank of $X_0$ and the number of subspaces are unknown either. Can we recover the \emph{row space} of $X_0$, and the identities of the non-zero columns of $C_0$, efficiently and exactly? If so, under which conditions?

This problem is motivated from the famous \emph{subspace segmentation} problem (Costeira and Kanade, 1998; Eldar and Mishali, 2009; Elhamifar and Vidal, 2009; Fischler and Bolles, 1981; Gear, 1998; Gruber and Weiss, 2004; Liu et al., 2010b,c; Rao et al., 2010; Vidal, 2011; Ma et al., 2007, 2008), as often in computer vision and image processing applications, one observes data points drawn from the union of \emph{multiple} subspaces. The goal of subspace segmentation is to segment the samples into their respective subspaces. In fact, subspace segmentation can be regarded as a generalization of Principal Component Analysis (PCA) that has only {\em one} subspace. As such, similar to PCA, segmentation algorithms can be sensitive to the presence of outliers. In fact, because of the coupling between segmentation and outlier detection, robust subspace segmentation appears to be a challenging problem not ever well studied in theory.

Interestingly, as we show below in Section \ref{sec:pre:rowspace}, the row space of the data samples $X_0$ determines the correct segmentation. Thus, both subspace segmentation and outlier detection can be transformed into solving problem~(\ref{equ.beginproblem}), where the column support of $C_0$ indicates the outliers, and the row space of $X_0$ gives the segmentation result of the ``authentic'' samples.  To solve problem \eqref{equ.beginproblem}, we analyze the following convex optimization problem, termed {\em Low-Rank Representation (LRR)} (Liu et al., 2010b):
\begin{eqnarray}\label{eq:lrr}
 \min_{Z,C}  ||Z||_*+\lambda{||C||_{2,1}},& \textrm{ s.t. }& X = XZ+C,
\end{eqnarray}
where $\norm{\cdot}_*$ denotes the sum of the singular values, also known as \emph{nuclear norm} (Fazel, 2002),  the trace norm or Ky Fan norm; $\norm{\cdot}_{2,1}$ is called the $\ell_{2,1}$ norm and defined as the sum of $\ell_2$ norms of the columns of a matrix, and the parameter $\lambda>0$ is used to balance the effects of the two parts.

Using nuclear-norm based approach to tackle the subspace segmentation problem is not a new idea. In Liu et al. (2010b), the authors showed that if there is no outlier, then a formulation
\begin{eqnarray*}\min_{Z} ||Z||_* ,&\textrm{s.t.}& X = XZ,
\end{eqnarray*}
exactly solves the subspace segmentation problem. They further conjectured that in the presence of corruptions, the formulation (\ref{eq:lrr}) may be helpful. However, no theoretic analysis was offered. In contrast, we show that under mild conditions, both the row space of $X_0$ and the column support of $C_0$ can be recovered by solving problem~(\ref{eq:lrr}). Thus, one can simultaneously perform subspace segmentation and outlier detection in an efficient way.

While our analysis shares similar features as previous work in Robust Principal Component Analysis (RPCA, e.g., Cand{\`e}s et al., 2009; Xu et al., 2010), it is considerably more challenging due to the fact that the variable $Z$ is left-multiplied by a dictionary matrix $X$, and that the dictionary itself is contaminated by outliers. Also, it is worth noting that the problem of recovering {\em row space} with column-wise corruptions essentially cannot be addressed by existing RPCA methods (Torre and Black, 2001; Xu et al., 2010), which are designed for recovering the {\em column space} with column-wise corruptions. In this regard, LRR also has a unique role in solving the RPCA problem under the context of corrupted features (i.e., row-wise corruptions); that is, one can recover the column space with row-wise corruptions by solving the following transposed version of \eqref{eq:lrr}:
\begin{eqnarray*}
\min_{Z,C} ||Z||_*+\lambda{||C||_{2,1}},& \textrm{s.t.} &X^T = X^TZ+C.
\end{eqnarray*}
As discussed above, existing RPCA methods (e.g., Xu et al., 2010) that focus on recovering the column space with column-wise corruption are fundamentally unable to address this problem.
\section{Preliminaries}\label{sec:notion_pre}
For easy of reading, we introduce in this section some preliminaries, including the usage of mathematical notations, the concept of independent subspaces, the role of row space in subspace segmentation, and some previous results about recovering row space by LRR.
\subsection{Summary of Notations}
Capital letters such as $M$ are used to represent matrices, and accordingly, $[M]_i$ denotes the $i$-th column vector of $M$. Letters $U$, $V$ , $\mathcal{I}$ and their variants (complements, subscripts,
etc.) are reserved for column space, row space and column support, respectively. There are four associated projection operators we use throughout. The projection onto the column space, $U$, is denoted by $\mathcal{P}_U$
and given by $\mathcal{P}_U(M)=UU^TM$, and similarly for the row space $\mathcal{P}_V(M)=MVV^T$. Sometimes, we need to apply $\mathcal{P}_V$ on the left side of a matrix. This special operator is denoted by $\mathcal{P}_V^L$ and given by $\mathcal{P}_V^L(\cdot)=VV^T(\cdot)$. The matrix $\mathcal{P}_{\mathcal{I}}(M)$ is obtained from $M$ by setting column $[M]_i$ to zero for all $i\not\in\mathcal{I}$. Finally, $\mathcal{P}_T$ is the projection to the space spanned by $U$ and $V$, and given by $\mathcal{P}_T(\cdot) = \mathcal{P}_U(\cdot)+\mathcal{P}_V(\cdot)-\mathcal{P}_U\mathcal{P}_V(\cdot)$. Note that $\mathcal{P}_T$ depends on both $U$ and
$V$, and we suppress this notation wherever it is clear which $U$ and $V$ we are using. The complementary operators, $\mathcal{P}_{U^{\bot}}$, $\mathcal{P}_{V^{\bot}}$, $\mathcal{P}_{T^{\bot}}$, $\mathcal{P}_{V^{\bot}}^{L}$ and $\mathcal{P}_{\mathcal{I}^{c}}$ are defined as usual (Xu et al., 2010). The same notation is also used to represent a subspace of matrices: e.g., we write $M\in\mathcal{P}_{U}$ for any matrix $M$ that satisfies $\mathcal{P}_{U}(M)=M$. Five matrix norms are used: $\norm{M}_*$ is the nuclear norm, $\norm{M}_{2,1}$ is the sum of $\ell_2$ norms of the columns $[M]_i$, $\norm{M}_{2,\infty}$ is the largest $\ell_2$ norm of the columns, and $\|M\|_F$ is the Frobenius norm. The largest singular value of a matrix (i.e., the spectral norm) is $\norm{M}$, and the smallest positive singular value is denoted by $\sigma_{min}(M)$. The only vector norm used is $\norm{\cdot}_2$, the $\ell_2$ norm. Depending on the context, $\Id$ is either the identity matrix or the identity operator, and $\mathbf{e}_i$ is the $i$-th standard basis vector.

In particular, letters $X$, $Z$, $C$ and their variants (complements, subscripts, etc.) are reserved for the data matrix (also the dictionary), coefficient matrix (in LRR) and outlier matrix, respectively. The SVD of $X_0$ and $X$ are $U_0\Sigma_0V_0^T$ and $U_X\Sigma_{X}V_X^T$, respectively. We use $\mathcal{I}_0$ to denote the column support of $C_0$, $d$ the ambient data dimension, $n$ the total number of data points in $X$, $\gamma\triangleq{}|\mathcal{I}_0|/n$ the fraction of outliers, and $r_0$ the rank of $X_0$. For a convex function $f:\mathcal{R}^{m\times{m'}}\rightarrow\mathbb{R}$, we say that $Y$ is a subgradient of $f$ at $M$, denoted as $Y\in\partial{}f(M)$, if and only if $f(M')\geq{}f(M)+\langle{}M'-M,Y\rangle,\forall{}M'$. We also adopt the conventions of using $\spn{M}$ to denote the linear space spanned by the columns of a matrix $M$, using $y\in\spn{M}$ to denote that a vector $y$ belongs to the space $\spn{M}$, and using $Y\in\spn{M}$ to denote that all column vectors of $Y$ belong to $\spn{M}$.
\subsection{Independent Subspaces}
The concept of independence will be used in our analysis. Its definition is as follows:
\begin{defn}
A collection of $k$ ($k\geq2$) subspaces $\{\mathcal{S}_1,\mathcal{S}_2,\cdots,\mathcal{S}_k\}$ are independent if and only if $\mathcal{S}_i\cap\sum_{j\neq{i}}\mathcal{S}_j=\{0\}$.
\end{defn}
There is a concept closely related to the independence, namely the \emph{pairwise disjoint} assumption, which holds if and only if $\mathcal{S}_i\cap\mathcal{S}_j=\{0\},\forall{i\neq{j}}$, i.e., there is no intersection between any two subspaces. While there are only two subspaces (i.e., $k=2$), independence is equivalent to pairwise disjointness. While $k>2$, independence is a sufficient condition for pairwise disjointness, but not necessary.
\subsection{Relation Between Row Space and Segmentation}\label{sec:pre:rowspace}
The subspace memberships of the authentic samples are  determined by the row space $V_0$. Indeed, as shown in Costeira and Kanade (1998) and Gear (1998), when subspaces are independent, $V_0V_0^T$ forms a block-diagonal matrix: the $(i,j)$-th entry of $V_0V_0^T$ can be non-zero only if the $i$-th and $j$-th samples are from the same subspace. Hence, this matrix, termed as {\em Shape Iteration Matrix} (SIM) (Gear, 1998), has been widely used for subspace segmentation (Costeira and Kanade, 1998; Gear, 1998; Vidal, 2011). Previous approaches simply compute the SVD of the data matrix $X=U_X\Sigma_{X}V_X^T$ and then use $|V_XV_X^T|$ for subspace segmentation. However, in the presence of outliers, $V_X$ can be far away from $V_0$ and thus the segmentation using such approaches may be inaccurate. In contrast, we show that LRR can recover $V_0V_0^T$ even when data matrix $X$ are corrupted by outliers.
\begin{figure}
\begin{center}
\vspace{0.1in}
\includegraphics[width=0.28\textwidth]{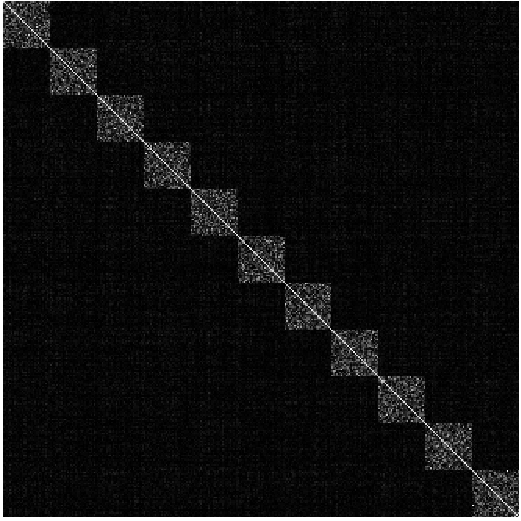}
\caption{An example of the matrix $V_0V_0^T$ computed from dependent subspaces. In this example, we create 11 pairwise disjoint subspaces each of which is of dimension 20, and draw 20 samples from each subspace. The ambient dimension is 200, which is smaller than the sum of the dimensions of the subspaces. So the subspaces are dependent and $V_0V_0^T$ is not strictly block-diagonal. Nevertheless, it is simple to see that high segmentation accuracy can be achieved by using the above similarity matrix to do spectral clustering.}\label{fig:sim}
\end{center}
\end{figure}

If the subspaces are not independent, $V_0V_0^T$ may not be block-diagonal. This is indeed well expected, since when the subspaces have nonzero (nonempty) intersections, then some samples may belong to multiple subspaces simultaneously. Nevertheless, when the subspaces are pairwise disjoint (but not independent), our extensive numerical experiments show that $V_0V_0^T$ is close to be block-diagonal, as exemplified in Figure \ref{fig:sim}. Hence, to recover $V_0V_0^T$ is still of interest to subspace segmentation. Note that the analysis in this work focuses on when $V_0V_0^T$ can be recovered, and hence does not rely on whether or not the subspaces are independent.
\subsection{Relation Between Row Space and LRR}
To better illustrate our intuition, we begin with the ``ideal'' case where there is no outlier in the data: i.e., $X=X_0$ and $C_0=0$. Thus, the LRR problem reduces to $\min_{Z}\norm{Z}_*\textrm{ s.t. }X_0=X_0Z$. As shown in (Liu et al., 2010a), this problem has a unique solution given by $Z^*=V_0V_0^T$, i.e., the solution of LRR identifies the row space of $X_0$ in this special case. Thus, when the data are contaminated by outliers, it is natural to consider problem \eqref{eq:lrr}. To see how LRR recovers the row space, we first establish the following lemma which can be simply deduced by Theorem 4.3 of Liu et al. (2010a).
\begin{lem}\label{lemma:lrr:solution}
For any optimal solution $(Z^*,C^*)$ to the LRR problem \eqref{eq:lrr}, we have that
$$Z^*\in{}\mathcal{P}_{V_X}^L,$$
where $V_X$ is the row space of $X$.
\end{lem}
The above lemma states that the optimal solution (with respect to the variable $Z$) to LRR always locates within the row space of $X$. This provides us an important clue on the conditions for recovering $V_0V_0^T$ by $Z^*$.
\section{Settings and Results}\label{sec:results}
In this section we show that, under mild assumptions, LRR can \emph{exactly} recover both the row space of $X_0$ (i.e., the true SIM that encodes the subspace memberships of the samples) and the columns support of $C_0$ (i.e., the identities of the outliers), from $X$, as we detail below.

While several articles, e.g., Cand{\`e}s and Recht (2009); Cand{\`e}s et al. (2009) and Xu et al. (2010), have proven that the nuclear norm regularized optimization problems are powerful in dealing with corruptions including missed observations and outliers, it is considerably more challenging to establish the success conditions of LRR. This is partly due to  the bilinear interaction between the corrupted matrix $X=X_0+C_0$ and the unknown $Z$ in the equation $X_0+C_0=(X_0+C_0)Z+C$, which is essentially a matrix recovery task under a {\em noisy dictionary}, a topic not studied in literature to the best of our knowledge.  Moreover, our goal is to recover {\em row space} from column-wise corruptions. This is a new task not addressed by previous RPCA and matrix recovery methods that mainly focus on recovering column space (e.g., Cand{\`e}s et al., 2009; Cand{\`e}s and Plan, 2010; Cand{\`e}s and Recht, 2009; J. Devlin and Kettenring, 1981; Torre and Black, 2001; Wright et al., 2009; Xu
et al., 2010), and hence calls for new analysis tools.
\subsection{Problem Settings}\label{sec:assumptions}
We discuss in this subsection three conditions sufficient for LLR to succeed. Note that these conditions also reveal how the outliers and samples are defined in LRR.
\subsubsection{A Necessary Condition for Exact Recovery}
Suppose $(Z^*,C^*)$ is an optimal solution to \eqref{eq:lrr}, then Lemma \ref{lemma:lrr:solution} concludes that the column space of $Z^*$ is a subspace of $V_X$. Hence, for $Z^*$ (or a part of $Z^*$) to exactly recover $V_0$, $V_0$ must be a subspace of $V_X$, i.e., the following is a necessary condition:
\begin{eqnarray}\label{eq:setting:2}
V_0\in\mathcal{P}_{V_X}^L.
\end{eqnarray}
Note that if there are outliers that exactly lie on the subspaces, then $X$ will contain more samples than $X_0$ and thus the above condition is violated. So this condition can avoid the degenerative cases where the outliers palm themselves as subspace members. To show how it can hold, we establish the following lemma which show that \eqref{eq:setting:2} can be satisfied when the outliers are independent to the samples.
\begin{lem}\label{lemma:setting2:sufficient}
If $\spn{C_0}$ and $\spn{X_0}$ are independent to each other, i.e., $\spn{C_0}\cap\spn{X_0}=\{0\}$, then \eqref{eq:setting:2} holds.
\end{lem}
\subsubsection{Relatively Well-Definedness}
\begin{figure}
\begin{center}
\includegraphics[width=0.45\textwidth]{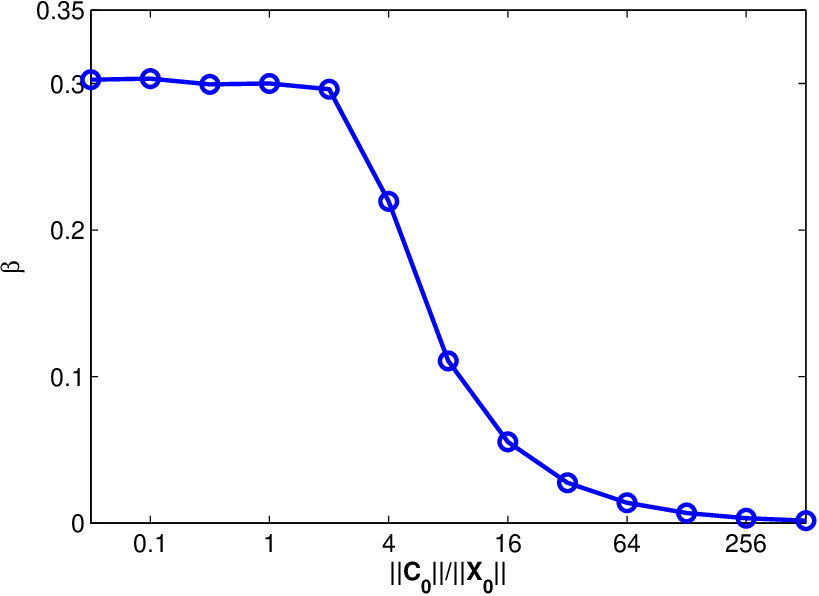}
\caption{Plotting the RWD parameter $\beta=1/(\|X\|\|\Sigma_X^{-1}V_X^TV_0\|)$ as a function of the relative magnitude $\|C_0\|/\|X_0\|$. These results are from our numerical experiments. In those experiments, the outlier fraction is fixed to be $\gamma=0.5$, and the outlier magnitude is varied for investigation. The matrices $X_0$ and $C_0$ are generated in a similar way as in Section \ref{sec:diss}.} \label{fig:pose}
\end{center}
\end{figure}
To reveal the success conditions of LRR, as mentioned, one technical challenge comes from the bilinear interaction between the corrupted matrix $X=X_0+C_0$ and the unknown $Z$ in the equation $X=XZ+C$. Actually, this issue also makes the equation $X=XZ+C$ ``seems'' questionable, because the data matrix $X$ (which itself contains outliers) is used as the dictionary for reconstruction. Nevertheless, we show that the success of LRR can be exactly ensured if $X$ satisfies the following \emph{relatively well-defined} (RWD) condition.
\begin{defn}
The dictionary $X$ generated by $X=X_0+C_0$, with SVD $X=U_X\Sigma_XV_X^T$ and $X_0=U_0\Sigma_0V_0^T$, is said to be RWD (with regard to $X_0$) with parameter $\beta$ if
\begin{eqnarray}\label{eq:beta}
\|\Sigma_X^{-1}V_X^TV_0\|\leq\frac{1}{\beta\|X\|}.
\end{eqnarray}
\end{defn}
To ensure the success of LRR, we require that the RWD parameter $\beta$ is not extremely small. If $X$ is perfectly well-defined (e.g., $r_0=1$ and $C_0=0$), then $\beta=1$. Without any assumptions, the above definition implies that $\beta$ is bounded by
\begin{eqnarray*}
\beta\geq{}\frac{1}{cond(X)},
\end{eqnarray*}
where $cond(X)=\|X\|/\sigma_{min}(X)$ is the condition number of $X$. This bound, however, does not guarantee the validity of RWD when $X$ is severely singular, i.e., $\sigma_{min}(X)\rightarrow0$ (this is possible in the presence of outliers). Fortunately, we show that the RWD parameter $\beta$ can be reasonably large under practical assumptions, e.g., the outlier magnitude is not extremely large. More precisely, we have the following lemma that estimates a lower bound of $\beta$.
\begin{lem}\label{lemma:beta:bound}
If $\spn{C_0}$ and $\spn{X_0}$ are independent to each other, then
\begin{eqnarray*}
\beta\geq{}\frac{\sin(\theta)}{cond(X_0)(1+\frac{\|C_0\|}{\|X_0\|})},
\end{eqnarray*}
where $cond(X_0)=\|X_0\|/\sigma_{min}(X_0)$ is the condition number of $X_0$, and $\theta>0$ is the smallest principal angle between $\spn{C_0}$ and $\spn{X_0}$.
\end{lem}
\begin{rmk}To ensure that $\beta$ is reasonably large, the above lemma suggests that the outlier magnitude should not be too large comparing to the sample magnitude. This is verified by our numerical experiments, as shown in Fig.\ref{fig:pose}.
\end{rmk}
\begin{rmk} To ensure that $\beta$ is reasonably large, the above lemma also suggests that the principal angle $\theta$ should be notably large; that is, the outliers in LRR are restricted to the data points which are notably far way from the underlying subspaces. This conclusion is consistent with the experimental observations reported in (Liu et al., 2010a), which shows that LRR can distinguish between the outliers and the corrupted samples, where a corrupted sample is drawn from the subspaces, but is corrupted to be away from the underlying subspaces.
\end{rmk}
\subsubsection{Incoherence}
Finally, as now standard (Cand{\`e}s and Recht, 2009; Cand{\`e}s et al., 2009; Xu et al., 2010), we require the \emph{incoherence condition} to hold, to avoid the issue of un-identifiability. As an extreme example, consider the case where the data matrix $X_0$ is non-zero in only one column. Such a matrix is both low-rank and column-sparse, thus the problem is unidentifiable. To make the problem meaningful, the low-rank matrix $X_0$ cannot itself be column-sparse. This is ensured via the following incoherence
condition.
\begin{defn} The matrix $X_0\in\mathbb{R}^{d\times{}n}$ with SVD $X_0=U_0\Sigma{}_0V_0^T$, $\rank{X_0}=r_0$ and $(1-\gamma)n$ of whose columns are non-zero, is said to be column-incoherent with parameter $\mu$ if
\begin{eqnarray}\label{eq:mu}
\max_{i}\|V_0^T\mathbf{e}_i\|^2\leq\frac{\mu{}r_0}{(1-\gamma)n},
\end{eqnarray}
where $\{\mathbf{e}_i\}$ are the standard basis vectors.
\end{defn}
Thus if $V_0$ has a column aligned with a coordinate axis, then $\mu = (1-\gamma)n/r_0$. Similarly, if $V_0$ is perfectly incoherent (e.g., if $r_0=1$ and every non-zero entry of $V_0$ has magnitude $1/\sqrt{(1-\gamma)n}$ ), then $\mu=1$.

\subsection{The Main Result}
Although the LRR problem \eqref{eq:lrr} may have multiple solutions, we show that any solution $(Z^*,C^*)$ to \eqref{eq:lrr} exactly recovers the row space of the low-rank matrix $X_0$, and the column support of $C_0$. The main result of this paper is shown in the following theorem.
\begin{thm}\label{theorem:lrr:outlier}
Suppose a given data matrix $X$ is generated by $X=X_0+C_0$, where $X_0$ is of rank $r_0$, $X$ has RWD parameter $\beta$ and $X_0$ has incoherence parameter $\mu$. Suppose $C_0$ is supported on $\gamma{}n$ columns. Let $\gamma^*$ be such that
\begin{eqnarray}\label{eq:gamma}
\frac{\gamma^*}{1-\gamma^*}=\frac{324\beta^2}{49(11+4\beta)^2\mu{}r_0},
\end{eqnarray}
then LRR with parameter $\lambda=\frac{3}{7\|X\|\sqrt{\gamma^*n}}$ strictly succeeds, as long as $\gamma\leq\gamma^*$ and \eqref{eq:setting:2} holds. Here, the success is in a sense that any optimal solution $(Z^*,C^*)$ to \eqref{eq:lrr} can produce
\begin{eqnarray}\label{eq:exact}
U^*(U^*)^T=V_0V_0^T&\textrm{and}&\mathcal{I}^*=\mathcal{I}_0,
\end{eqnarray}
where $U^*$ is the column space of $Z^*$, and $\mathcal{I}^*$ is column support of $C^*$.
\end{thm}

The performance (i.e., the value of $\gamma^*$) of LRR depends on the properties of data, mainly including the rank $r_0$ (the lower the better), the RWD parameter $\beta$ (the larger the better), and the incoherence parameter $\mu$ (the smaller the better). Interestingly, the above theorem also implies that LRR can be used to solve a challenging PCA problem (which is presented in the Introduction), which is to recover the column space with corrupted features (i.e., row-wise corruption). By solving the transposed version of LRR (see the Introduction), the row space of $X_0^T$, i.e., the column space of $X_0$, can be recovered.
\subsection{Proof Outline}
In this section we provide an outline for the proof of Theorem \ref{theorem:lrr:outlier}. The full proof appears in the appendix section. The proof follows three main steps.
\begin{itemize}
\item[1.] Identify the necessary and sufficient conditions (called equivalent conditions), for any pair $(Z',C')$ to produce the exact results \eqref{eq:exact}.
\item[2.] For a candidate pair $(Z',C')$ that respectively has the desired row space and column support, identify the sufficient conditions for $(Z',C')$ to be an optimal solution to the LRR problem \eqref{eq:lrr}. These conditions are called \emph{dual conditions}.
\item[3.] Show that the dual conditions can be satisfied, i.e., construct the \emph{dual certificates}.
\end{itemize}

\textbf{Equivalent Conditions:} For any feasible pair $(Z',C')$ that satisfies $X=XZ'+C'$, let the SVD of $Z'$ as $U'\Sigma'V'^T$ and the column support of $C'$ as $\mathcal{I}'$. In order to produce the exact results \eqref{eq:exact}, on the one hand, a necessary condition is that $\mathcal{P}_{V_0}^L(Z')=Z'$ and $\mathcal{P}_{\mathcal{I}_0}(C')=C'$, as this is nothing but $U'$ is a subspace of $V_0$ and $\mathcal{I}'$ is a subset of $\mathcal{I}_0$. On the other hand, it can be proven that $\mathcal{P}_{V_0}^L(Z')=Z'$ and $\mathcal{P}_{\mathcal{I}_0}(C')=C'$ are sufficient to ensure $U'U'^T=V_0V_0^T$ and $\mathcal{I}'=\mathcal{I}_0$. So, the exactness described in \eqref{eq:exact} can be equally transformed into two constraints: $\mathcal{P}_{V_0}^L(Z')=Z'$ and $\mathcal{P}_{\mathcal{I}_0}(C')=C'$, which we will use to construct an oracle problem to facilitate the proof.

\textbf{Dual Conditions:} For the pair $(Z',C')$ that satisfies $X=XZ'+C'$, $\mathcal{P}_{V_0}^L(Z')=Z'$ and $\mathcal{P}_{\mathcal{I}_0}(C')=C'$, let the SVD of $Z'$ as $U'\Sigma'V'^T$ and the column-normalized version of $C'$ as $H'$. That is, column $[H']_i=\frac{[C']_i}{\|[C']_i\|_2}$ for all $i\in\mathcal{I}_0$, and $[H']_i=0$ for all $i\not\in\mathcal{I}_0$ (note that the column support of $C'$ is $\mathcal{I}_0$). Furthermore, define $\mathcal{P}_{T'}(\cdot)=\mathcal{P}_{U'}(\cdot)+\mathcal{P}_{V'}(\cdot)-\mathcal{P}_{U'}\mathcal{P}_{V'}(\cdot)$. With these notations, it can be proven that $(Z',C')$ is an optimal solution to LRR if there exists a matrix $Q$ that satisfies
\begin{eqnarray*}
\mathcal{P}_{T'}(X^TQ)=U'V'^T,&&\|X^TQ-\mathcal{P}_{T'}(X^TQ)\|<1\\
\mathcal{P}_{\mathcal{I}_0}(Q)=\lambda{}H',&&\|Q-\mathcal{P}_{\mathcal{I}_0}(Q)\|_{2,\infty}<\lambda.
\end{eqnarray*}
Although the LRR problem \eqref{eq:lrr} may have multiple solutions, it can be further proven that any solution has the desired row space and column support, provided the above conditions have been satisfied. So, the left job is to prove the above dual conditions, i.e., construct the dual certificates.

\textbf{Dual Certificates:} The construction of dual certificates mainly concerns a matrix $Q$ that satisfies the dual conditions. However, since the dual conditions also depend on the pair $(Z',C')$, we actually need to obtain three matrices, $Z'$, $C'$ and $Q$. This is done by considering an alternate optimization problem, often called the ``oracle problem''. The oracle problem arises by imposing the success conditions as additional constraints in \eqref{eq:lrr}:
\begin{eqnarray*}
\min_{Z,C} && \|Z\|_*+\lambda\|C\|_{2,1}\\\nonumber
\textrm{s.t.}&&X=XZ+C, \mathcal{P}_{V_0}^L(Z)=Z, \mathcal{P}_{\mathcal{I}_0}(C)=C.
\end{eqnarray*}
While it is not practical to solve the oracle problem since $V_0$ and $\mathcal{I}_0$ are both unknown, it significantly facilitate our proof.
Note that the above problem is always feasible, as $(V_0V_0^T,C_0)$ is feasible. Thus, an optimal solution, denoted as $(\hat{Z},\hat{C})$, exists. Observe that because of the two additional constraints, $(\hat{Z},\hat{C})$ satisfies \eqref{eq:exact}. Therefore, to show Theorem~\ref{theorem:lrr:outlier} holds, it suffices to show that $(\hat{Z},\hat{C})$ is the optimal solution to LRR.
 With this perspective, we construct the dual certificates using $(\hat{Z},\hat{C})$. Let the SVD of $\hat{Z}$ be $\hat{U}\hat{\Sigma}\hat{V}^T$, and the column-normalized version of $\hat{C}$ be $\hat{H}$. It is easy to see that there exists an orthonormal matrix $\bar{V}$ such that $\hat{U}\hat{V}^T=V_0\bar{V}^T$, where $V_0$ is the row space of $X_0$. Moreover, it is easy to show that $\mathcal{P}_{\hat{U}}(\cdot)=\mathcal{P}_{V_0}^L(\cdot)$, $\mathcal{P}_{\hat{V}}(\cdot)=\mathcal{P}_{\bar{V}}(\cdot)$, and hence the operator $\mathcal{P}_{\hat{T}}$ defined by $\hat{U}$ and $\hat{V}$, obeys $\mathcal{P}_{\hat{T}}(\cdot)=\mathcal{P}_{V_0}^L(\cdot)+\mathcal{P}_{\bar{V}}(\cdot)-\mathcal{P}_{V_0}^L\mathcal{P}_{\bar{V}}(\cdot)$.  Finally, the dual certificates are finished by constructing $Q$ as follows:
\begin{eqnarray*}
Q_1&\triangleq&\lambda\mathcal{P}_{V_0}^L(X^T\hat{H}),\\
Q_2&\triangleq&\lambda\mathcal{P}_{V_0^{\bot}}^L\mathcal{P}_{\mathcal{I}_0^c}\mathcal{P}_{\bar{V}}(\Id+\sum_{i=1}^{\infty}(\mathcal{P}_{\bar{V}}\mathcal{P}_{\mathcal{I}_0}\mathcal{P}_{\bar{V}})^i)\mathcal{P}_{\bar{V}}(X^T\hat{H}),\\
Q&\triangleq&U_X\Sigma_X^{-1}V_X^T(V_0\bar{V}^T+\lambda{}X^T\hat{H}-Q_1-Q_2),
\end{eqnarray*}
where $U_X\Sigma_XV_{X}^T$ is the SVD of the data matrix $X$.
\section{Experiments}\label{sec:diss}
Notice that LRR have been used to achieve state-of-the-art performances in several applications such as motion segmentation (Liu et al., 2010a; Liu and Yan, 2011; Favaro et al.,
2011), image segmentation (Chen et al., 2011), saliency detection (Lang et al., 2011) and face recognition (Liu and Yan, 2011). In particular, motion segmentation and image segmentation are typical examples of the subspace segmentation problem. Also, by using appropriate visual features to describe the images, saliency detection can be casted into an example of the outlier detection problem, as shown in (Lang et al., 2011). So, there have been extensive experiments to verify the effectiveness of LRR. Here, we shall further show some experimental results to verify the theoretical results obtained in this paper.
\begin{figure*}
\begin{center}
\includegraphics[width=0.7\textwidth]{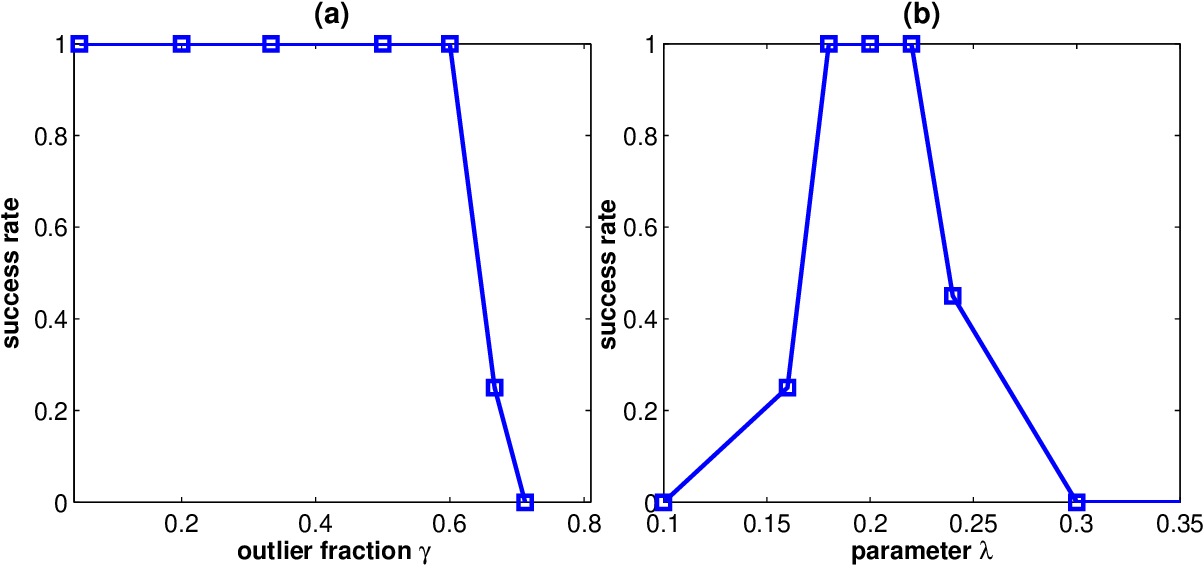}
\caption{The success rates obtained from 50 random trials. (a) When $\lambda=0.2$, the success rates obtained under various settings of the outlier fraction $\gamma$. (b) When the outlier fraction is fixed to be $\gamma=0.5$, plotting the success rate as a function of the parameter $\lambda$. In these experiments, the ``success'' is measured in terms of exact recovery, i.e., $U^*(U^*)^T=V_0V_0^T$ and $\mathcal{I}^*=\mathcal{I}_0$. } \label{fig:toy}
\end{center}
\end{figure*}
\subsection{Numerical Results}
Theorem \ref{theorem:lrr:outlier} states that there exists a parameter $\lambda$ such that LRR can work well while the outlier fraction is not larger than a certain threshold. To explore this, we construct 5 pairwise disjoint subspaces $\{\mathbb{S}_i\}_{i=1}^5$ whose bases $\{U_i\}_{i=1}^{5}\in\mathbb{R}^{500}$ are computed by $U_{i+1} = TU_{i},1\leq{}i\leq{4}$, where $T$ is a random rotation and $U_1$ is a random orthonomal matrix of dimension $500\times 5$. So, each subspace is of dimension 5. We sample 40 data samples from each subspace by $X_i=U_iR_i,1\leq{}i\leq{5}$ with $R_i$ being a $5\times{40}$ uniform matrix with a range from -1 to 1, and construct the sample matrix $X_0$ as $X_0=[X_1,\cdots,X_5]$. Some outliers are randomly generated from zero mean Gaussian distribution with standard deviation $s$, where $s$ is set to be the averaged absolute value of the samples, i.e., the samples and outliers approximately have the same magnitude.

While fixing all the other configurations, we change the number of outliers and the parameter $\lambda$. Then we observe whether the recovery is exact or not. Here, the exactness is in a sense that $\|U^*(U^*)^T-V_0V_0^T\|<10^{-4}$ (i.e., $U^*(U^*)^T=V_0V_0^T$), and $\mathcal{I}^*=\mathcal{I}_0$ with $\mathcal{I}^*=\{i:\|[C^*]_i\|_2\geq10^{-4}\|[X]_i\|_2\}$. Figure \ref{fig:toy}(a) shows that LRR can be exactly successful while $\gamma\leq0.6$, and Figure \ref{fig:toy}(b) illustrates that there exists a parameter range for obtaining exact recovery. These results are consistent with the statements in Theorem \ref{theorem:lrr:outlier}.
\subsection{Results on Real Data}
\begin{figure}
\begin{center}
\centerline{\includegraphics[width=0.45\textwidth]{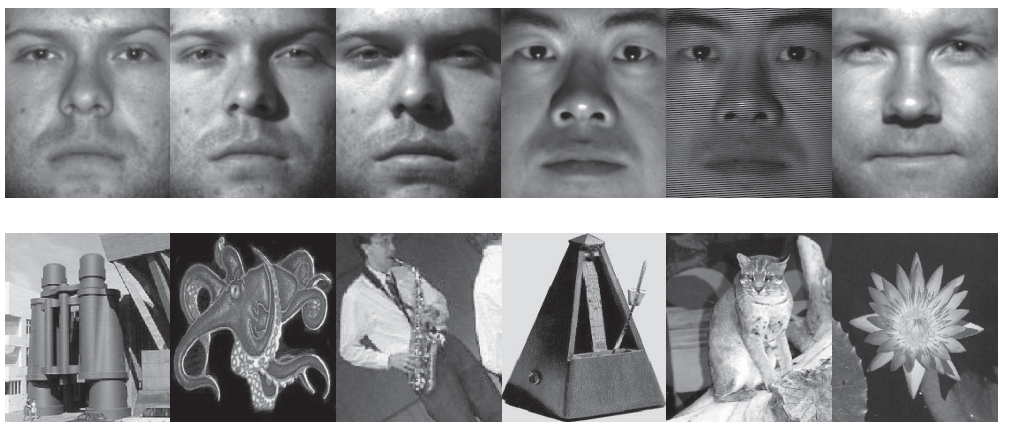}} \caption{Examples of the images in the Yale-Caltech dataset.} \label{fig:yaleb-cal:example}
\end{center}
\end{figure}
\subsubsection{Datasets} To test LRR's effectiveness in the presence of outliers and noise, we create a dataset by combing Extended Yale Database B (Lee et al., 2005) and Caltech101 (Li et al., 2004), so called as ``Yale-Caltech''. For Extended Yale Database B, we remove the images pictured under extreme light conditions. Namely, we only use the images with view directions smaller than 45 degrees and light source directions smaller than 60 degrees, resulting in 1204 authentic samples approximately drawn from a union of 38 low-rank subspaces (each face class corresponds to a subspace). For Caltech101, we only select the classes containing no more than 40 images, resulting in 609 non-face outliers. Fig.\ref{fig:yaleb-cal:example} shows some examples of this dataset.\\
\subsubsection{Evaluation Metrics}
\textbf{Segmentation Accuracy (ACC):} The segmentation results can be evaluated in a similar way as classification results. Nevertheless, since segmentation methods cannot provide the class label for each cluster, a postprocessing step is needed to assign each cluster a label: given the ground truth classification results, the label of a cluster is the index of the ground truth class that contributes the maximum number of samples to the cluster. Then, we compute the segmentation accuracy (ACC) as the percentage of correctly classified samples.\\
\\\textbf{Areas Under Curve (AUC):} As shown in Theorem \ref{theorem:lrr:outlier}, the minimizer $C^*$ (with respect to the variable $C$) can be used to detect the outliers that possibly exist in data. This can be simply done by finding the nonzero columns of $C^*$, when all or a fraction of data samples are clean. For the cases where the data is noisy and the learnt $C^*$ only approximately has sparse column supports, one could use thresholding strategy; that is, the $i$-th data vector of $X$ is judged to be outlier if and only if
\begin{eqnarray*}
\|[C^*]_{:,i}\|_2>\delta,
\end{eqnarray*}
where $\delta>0$ is a parameter. To evaluate the effectiveness of outlier detection without choosing a parameter $\delta$, we consider the receiver operator characteristic (ROC), which is widely used to evaluate the performance of binary classifiers. The ROC curve is obtained by trying all possible thresholding values, and for each value, plotting the true positives rate on the Y-axis against the false positive rate value on the X-axis. The areas under the ROC curve, known as AUC, provides a number for evaluating the quality of outlier detection. Note that the AUC score is the larger the better, and always ranges between 0 and 1.
\subsubsection{Results}
\begin{table}[t]
\caption{Segmentation accuracy (ACC) and AUC comparison on the Yale-Caltech dataset.}
\label{tb:yale}
\begin{center}
\begin{tabular}{|ccccc|}\hline
                &PCA            &RPCA$_1$       &RPCA$_{2,1}$     &LRR\\\hline
ACC (\%)        &77.15          &82.97          &83.72            &\textbf{86.13}\\
AUC             &0.9653         &0.9819         &0.9863           &\textbf{0.9927}\\
\hline
\end{tabular}
\end{center}
\end{table}
The goal of this test is to identify 609 non-face outliers and segment the rest 1204 face images into 38 clusters. The performance of segmentation and outlier detection is evaluated by ACC and AUC, respectively. While investigating segmentation performance, the affinity matrix is computed from all images, including both the face images and non-face outliers. However, for the convenience of evaluation, the outliers and the corresponding affinities are removed (according to the ground truth) before obtaining the segmentation results.

We resize all images into $20\times20$ pixels and form a data matrix $X$ of size $400\times1813$. Table \ref{tb:yale} shows the results of PCA, RPCA$_1$ (Cand{\`e}s et al., 2009), RPCA$_{2,1}$ (Xu et al., 2010) and LRR. It can be seen that LRR is better than PCA and RPCA methods, in terms of both subspace segmentation and outlier detection. Here, the advantages (in terms of subspace segmentation) of LRR are mainly due to its methodology. More precisely, LRR \emph{directly} targets on recovering the row space $V_0V_0^T$, which provably determines the segmentation results. In contrast, PCA and RPCA methods target on recovering the column space $U_0U_0^T$, which is designed for dimension reduction. This is why LRR are better than PCA and RPCA methods as a tool for subspace segmentation. In terms of outlier detection, LRR's advantages mainly come from the fact that this dataset has a structure of multiple subspaces, which fits well the assumptions of LRR (Liu et al., 2010a,b). Whereas, PCA and RPCA methods are based on the assumption that the data is sampled from a single subspace. When the data is drawn from a union of multiple subspaces, PCA and RPCA methods actually treat those multiple subspaces as a single one. Since the specifics of the individual subspaces are not well considered, they may lose some accuracy in the detection of outliers.
\section{Conclusion}\label{sec:conclusion}
This paper studies the problem of subspace segmentation in the presence of outliers. We analyzed a convex formulation termed LLR, and showed that the optimal solution exactly recovers the row space of the authentic data and identifies the outliers. Since the row space determines the segmentation of data, LRR can perform subspace segmentation and outlier identification simultaneously.

The analysis presented in this paper differs from previous work (e.g., Cand{\`e}s et al., 2009; Xu et al., 2010) largely due to the fact that the dictionary used in \eqref{eq:lrr} is the data matrix $X$, as opposed to the (arguably easier) identity matrix $\Id$ used in (Cand{\`e}s et al., 2009; Xu et al., 2010). As a future direction, it is interesting to see whether the technique presented can be extended to general dictionary matrices other than $X$ or $\Id$.
\section*{References}
Emmanuel Cand{\`e}s and Yaniv Plan (2010). Matrix completion with noise. In \emph{Proceeding of IEEE}, {\bf 98}: 925-936.

Emmanuel Cand{\`e}s and Benjamin Recht (2009). Exact matrix completion via convex optimization.
\emph{Foundations of Computational Mathematics}, {\bf 9}(6):717-772.

Emmanuel Cand{\`e}s, Xiaodong Li, Yi~Ma and John Wright (2009). Robust principal component analysis? \emph{Journal of the ACM, to appear}.

Bing Chen, Guangcan Liu, Zhongyang Huang and Shuicheng Yan (2011). Multi-task low-rank affinities pursuit for image segmentation. In \emph{IEEE International Conference on Computer Vision}.

Paulo Costeira and Takeo Kanade (1998). A multibody factorization method for independently moving objects. \emph{International Journal on Computer Vision}, {\bf 29}(3): 159-179.

Yonina Eldar and Moshe Mishali (2009). Robust recovery of signals from a structured union of subspaces. \emph{IEEE Transactions on Information Theory}, {\bf 55}(11): 5302-5316.

E.~Elhamifar and Ren{\'e} Vidal (2009). Sparse subspace clustering. In \emph{IEEE Conference on Computer Vision and Pattern Recognition}, {\bf 2}: 2790-2797.

Paolo Favaro, Ren{\'e} Vidal and Avinash Ravichandran (2011). A closed form solution to robust subspace estimation and clustering. In \emph{IEEE Conference on Computer Vision and Pattern Recognition}.

M.~Fazel (2002). Matrix rank minimization with applications. \emph{PhD thesis}.

Fei-Fei Li, Rob Fergus and Pietro Perona (2004). Learning generative visual models from few training examples: An
  incremental bayesian approach tested on 101 object categories. In \emph{Workshop of CVPR}.

Martin Fischler and Robert Bolles (1981). Random sample consensus: a paradigm for model fitting with applications to image analysis and automated cartography.
\emph{Commun. ACM}, {\bf 24}(6): 381-395.

Charles William Gear (1998). Multibody grouping from motion images. \emph{International Journal on Computer Vision}, {\bf 29}(2): 133-150.

Amit Gruber and Yair Weiss (2004). Multibody factorization with uncertainty and missing data using the {EM} algorithm. In \emph{IEEE Conference on Computer Vision and Pattern Recognition},
  {\bf 1}: 707-714.

S.J. Devlin, R. Gnanadesikan and J.R. Kettenring (1981). Robust estimation of dispersion matrices and principal components. \emph{Journal of the American Statistical Association}, {\bf 76}(374): 354-362.

Andrew Knyazev and Merico Argentati (2002). Principal angles between subspaces in an A-based scalar product: Algorithms and perturbation estimates. \emph{SIAM J. Sci. Comput}, {\bf 23}: 2009--2041.

Fernando Torre and Michael Black (2001). Robust principal component analysis for computer vision. In \emph{IEEE International Conference on Computer Vision}, 362-369.

Congyan Lang, Guangcan Liu, Jian Yu and Shuicheng Yan (2011). Saliency detection by multi-task sparsity pursuit.
\newblock \emph{IEEE Transactions on Image Processing, to appear}.

Kuang-Chih Lee, Jeffrey Ho and David Kriegman (2005). Acquiring linear subspaces for face recognition under variable lighting. \emph{IEEE Trans. Pattern Anal. Mach. Intell.}, {\bf 27}(5): 684-698.

Guangcan Liu and Shuicheng Yan (2011). Latent low-rank representation for subspace segmentation and feature
  extraction. In \emph{IEEE International Conference on Computer Vision}.

Guangcan Liu, Zhouchen Lin, Shuicheng Yan, Ju Sun, Yi Ma and Yong Yu (2010a). Robust recovery of subspace structures by low-rank representation.
\newblock \emph{Preprint}.

Guangcan Liu, Zhouchen Lin and Yong Yu (2010b). Robust subspace segmentation by low-rank representation. In \emph{International Conference on Machine Learning (ICML)}, 663-670.

Guangcan Liu, Zhouchen Lin, Yong Yu and Xiaoou Tang (2010c). Unsupervised object segmentation with a hybrid graph model ({HGM}). \emph{IEEE Transactions on Pattern Analysis and Machine
  Intelligence}, {\bf 32}(5): 910-924.

Yi Ma, Harm Derksen, Wei Hong and John Wright (2007). Segmentation of multivariate mixed data via lossy data coding and compression. \emph{IEEE Transactions on Pattern Analysis and Machine
  Intelligence}, {\bf 29}(9): 1546-1562.

Yi Ma, Allen Yang, Harm Derksen and Robert Fossum (2008). Estimation of subspace arrangements with applications in modeling and
  segmenting mixed data. \emph{SIAM Review}, {\bf 50}(3): 413-458.

Shankar Rao, Roberto Tron, Ren{\'e} Vidal and Yi Ma (2010). Motion segmentation in the presence of outlying, incomplete, or corrupted trajectories. \emph{IEEE Transactions on Pattern Analysis and Machine
  Intelligence}, {\bf 32}(10): 1832-1845.

R.Tyrrell Rockafellar (1970). \emph{Convex Analysis}. Princeton University Press, NJ, USA.

Ren{\'e} Vidal (2011). Subspace clustering. \emph{IEEE Signal Processing Magazine}, {\bf 28}(2): 52-68.

John Wright, Arvind Ganesh, Shankar Rao, Yigang Peng and Yi Ma (2009). Robust principal component analysis: Exact recovery of corrupted
  low-rank matrices via convex optimization. In \emph{Advances in Neural Information Processing Systems}, 2080-2088.

Huan Xu, Constantine Caramanis and Sujay Sanghavi (2010). Robust PCA via outlier pursuit. In \emph{Advances in Neural Information Processing Systems}.
\end{document}